\documentclass[12pt]{article}
\usepackage{amsmath}
\usepackage{amsfonts}
\usepackage{a4,fullpage}

\usepackage[pdftex]{graphicx}
\DeclareGraphicsExtensions{.pdf,.jpeg,.png}
\newtheorem{definition}{theorem}
\newtheorem{lemma}{theorem}
\newtheorem{property}{theorem}
\newtheorem{theorem}{theorem}
\newenvironment{proof}{{\bf Proof}:}{\hfill $\Box$ }

\begin{document}

\title{Performance improvement of an optical network providing services based on multicast}
\author{Vincent Reinhard$^{1,2}$ Johanne Cohen$^{2}$, Joanna Tomasik$^{1}$,\\ Dominique Barth$^{2}$, Marc-Antoine Weisser$^{1}$ \\ \\
 (1) SUPELEC Systems Sciences, Computer Science Dpt., \\ 91192 Gif sur Yvette, France \\
email : FistName.LastName@supelec.fr\\ \   \\
(2) PRiSM, University of Versailles, \\ 45 avenue des Etats-Unis, 78035 Versailles, France\\
email : FistName.LastName@prism.uvsq.fr}

\maketitle

\begin{abstract}
Operators of networks covering large areas are confronted with demands from some of their customers who are virtual service providers. These providers may call for the connectivity service which fulfils the specificity of their services, for instance a multicast transition with allocated bandwidth. On the other hand, network operators want to make profit by trading the connectivity service of requested quality to their customers and to limit their infrastructure investments (or do not invest anything at all).

We focus on circuit switching optical networks and work on repetitive multicast demands whose source and destinations are {\em \`a priori} known by an operator. He may therefore have corresponding trees ``ready to be allocated'' and adapt his network infrastructure according to these recurrent transmissions. This adjustment consists in setting available branching routers in the selected nodes of a predefined tree.  The branching nodes are opto-electronic nodes which are able to duplicate data and retransmit it in several directions. These nodes are, however, more expensive and more energy consuming than transparent ones.

In this paper we are interested in the choice of nodes of a multicast tree where the limited number of branching routers should be located in order to minimize the amount of required bandwidth. After formally stating the problem we solve it by proposing a polynomial algorithm whose optimality we prove. We perform exhaustive computations to show an operator gain obtained by using our algorithm. These computations are made for different methods of the multicast tree construction. We conclude by giving dimensioning guidelines and outline our further work.
\end{abstract}

\section{Introduction}
Optical networks have become a dominant technology in modern networks covering large areas. Their advantage consists in providing an ultra-high bit rate obtained with slight energy consumption.   All-optical networks are particularly interesting from economic and ecological point of view because a cost of transparent routers is low and their energy consumption is negligible~\cite{BonChi+09}.

Modern networks face a growing demand on the part of service providers. New offered services are more complex than the simple connectivity service assured traditionally by network operators.  Providers sell services like storage and computation together with connectivity service to their customers. The part of this market ensuring on-the-fly resource allocation, called for commercial reasons Cloud Computing~\cite{ArmFox+10}, is under a rapid development. In order to meet the demands of their customers, virtual service providers have to purchase a guaranteed connectivity service at network operators.   At the same time, network operators can deal with numerous virtual service providers. They are interested in using their network resources the most efficiently and in this way minimize the cost of a prospective extension of their existing infrastructure.

We studied the mechanisms to execute distributed applications in an optical mesh network in the context of the CARRIOCAS project~\cite{Aud07,Ver+08}. Unlike a customary approach applied in Grids where applications benefit from a dedicated network infrastructure~\cite{FosKes+02}, this project went into the study of the coexistence of massive distributed applications in a network whose operator should make financial profit. With GMPLS~\cite{Man04} deployed, the CARRIOCAS network has to ensure both unicast and multicast transmissions. Routers which are able to duplicate data and send it in several directions allow a network operator to lower the bandwidth amount necessary to construct a multicast tree. On the other hand, these branching nodes are more expensive and more energy consuming than the transparent ones. The realistic assumption is thus that only a subset of routers is equipped with the duplicating functionality. In~\cite{ReiTom+09} we presented our solution to the problem consisting in the construction of a tree to any multicast request  with minimization of the amount of used bandwidth  under assumption of a limited number of branching nodes. The solution is heuristic because we proved that this problem is {\em NP}-complete. It turned out to be the most effective when the branching nodes were placed in the most homogeneous way in a network.   The most homogeneous placement of $k$ branching nodes represents in fact a solution to the $k$-centre problem which is also {\em NP}-complete~\cite{GarJoh79}.

Our study mentioned above inspired us to explore certain special cases of multicast demands. A network operator can know in advance recurrent multicast transmissions which require a lot of bandwidth. Being aware of frequent demands for identical (or almost identical) multicast transmissions an operator may have corresponding trees ``ready to be allocated'' and adapt his network infrastructure according to these recurrent transmissions. This adjustment may consist in setting available branching routers in the selected nodes of the predefined tree.  In this paper we are interested in the choice of nodes of a multicast tree where the branching routers should be located in order to minimize the amount of required bandwidth. This approach allows an operator to make his network more efficient without any additional cost.

In the following section we make a survey of existing solutions to multicast tree allocation   and explain the specificity of branching routers. In Section~\ref{PbForm} our problem is stated in the formal way. We also formulate (Section~\ref{properties}) the solution properties. Next, we propose an algorithm to solve our problem, compute its complexity, and prove that it gives an optimal solution. Our problem is evidenced to be polynomial.  Section~\ref{Results} presents the results of bandwidth requirements for multicast trees depending on the number of available branching routers. The multicast trees which are subject of this analysis have been obtained by two methods, the first one based on the shortest path approach and the second one based on the Steiner tree approach.  In the final section we give the conclusions and outline our further work.

\section{Multicast tree construction \label{TreeConstr}}

There are several schemes for multicasting data in networks~\cite{SalRee+97,JeoQia+00}. We present here the schemes adapted to optical circuit switching networks.  The first one is to construct virtual circuits from the multicast source to each destination. Such a scheme is equivalent to multiple unicasts (Fig.~\ref{MulticastAsUnicast}) and the network bandwidth used by a large multicast group may become unacceptable~\cite{MalZha+98}.
\begin{figure}
\begin{minipage}[b]{.46\linewidth}
\begin{center}
\includegraphics[height=1.2in,width=2.0in]{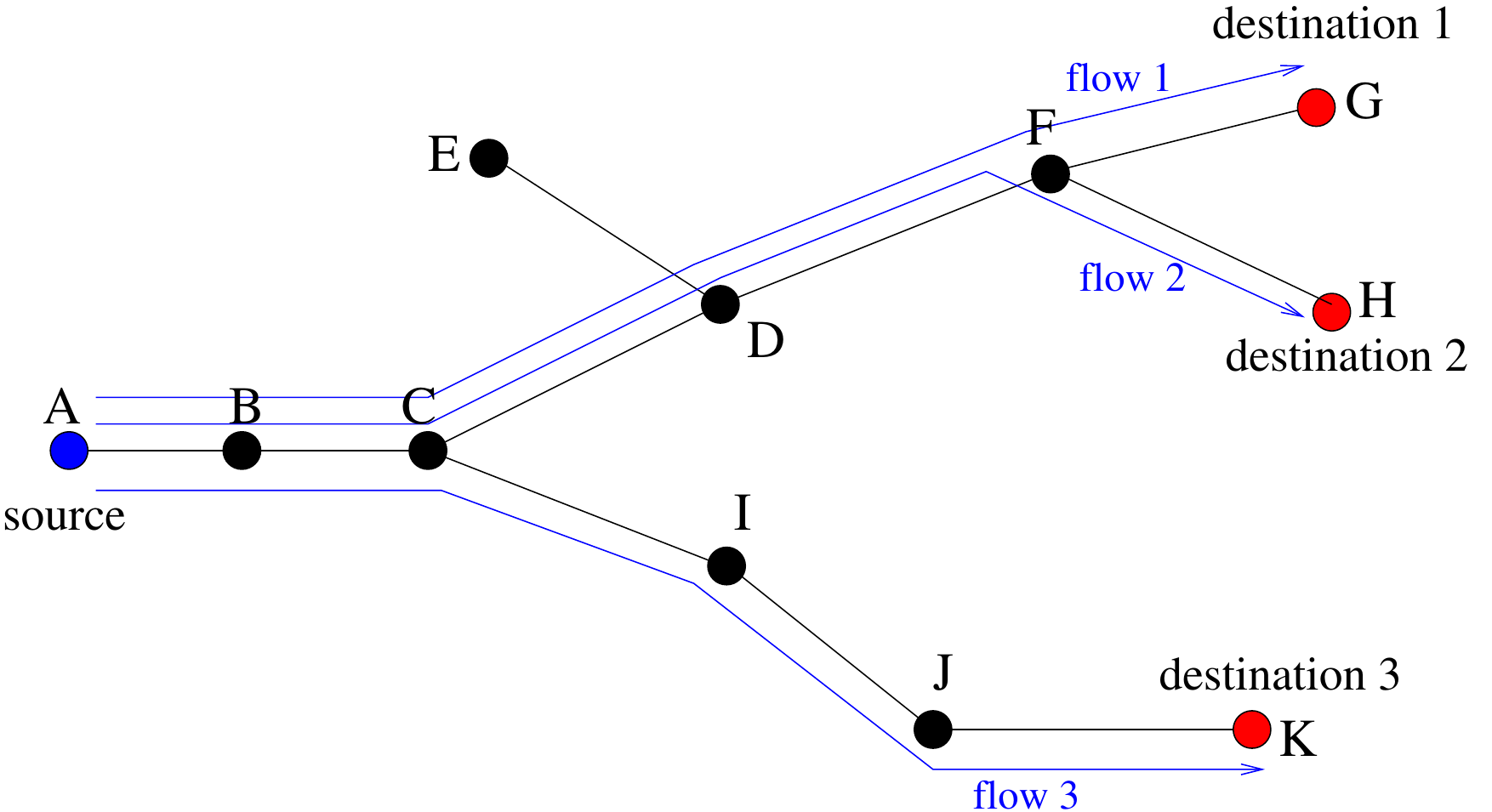}
\caption{A multicast with source $A$ and destinations $G$, $H$, $K$ built up as a set of unicasts (without branching nodes) \label{MulticastAsUnicast}}
\end{center}
 \end{minipage} \hfill
 \begin{minipage}[b]{.46\linewidth}
\begin{center}
\includegraphics[height=1.9in,width=2.0in]{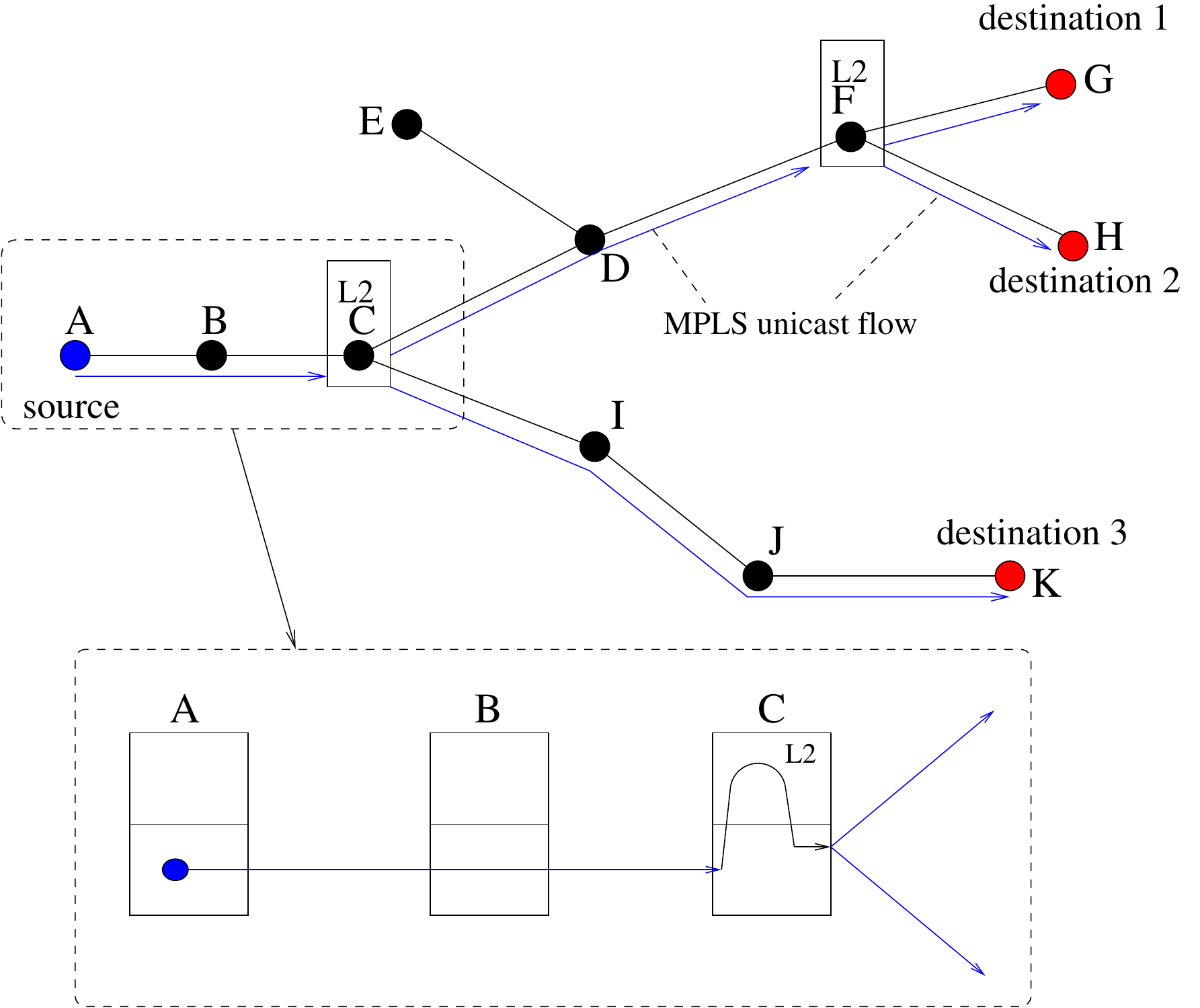}
\caption{A multicast with source $A$, destinations $G$, $H$, $K$ and branching nodes $C$, $F$ \label{MulticastWithBranching}}
\end{center}
 \end{minipage}
\end{figure}

In another scheme the multicast source sends data to the first destination and each destination acts as a source for the next destination until all destinations receive the data flow. In yet another scheme, intermediate routers make copies of data packets and dispatch them to their successors in the multicast tree. This solution allows the multicast transmission to share bandwidth on the common links. Numerous multicast tree algorithms, which follow the latter scheme, have been proposed and can roughly be classified into two categories~\cite{SalRee+97}. The first category contains the algorithms based on the shortest path while minimizing the weight of the path from the multicast source to each destination. The second category contains algorithms based on the Steiner tree problem~\cite{Bea89,BhaJaf83,HwaRic92,KomPas+92} which we formally define in Section~\ref{PbForm}. The algorithms derived from the Steiner tree problem minimize the total weight of the multicast tree. They are heuristic because the Steiner tree problem is NP-complete~\cite{HwaRic92}.

From the technological point of view, routers able to duplicate packets introduce a supplementary delay due to O/E/O conversions and are more expensive. For these reasons network operators want to limit the number of such routers which we call "diffusing nodes" or ``branching nodes''. The diffusing nodes which we consider are not equipped with the functionality ``drop-and-continue''~\cite{ZhaWei+00} as this operation mode is nowadays applied in practice exclusively in border routers. In Fig.~\ref{MulticastWithBranching} we go back to the example illustrated in Fig.~\ref{MulticastAsUnicast}. This time there are two branching nodes which allow one to reduce the amount of used bandwidth. Contrary to the solution built up of unicasts, in the one with branching nodes the bandwidth is used only once in each link.

\section{Formalization of optimization problem \label{PbForm}}
An optical network is modelled by a directed connected symmetrical graph~\cite{Ber66} $G=(V,E)$.  A multicast request is a pair $\epsilon = (e,R)$, where $e\in V$ is a multicast source and $R\subset V$ is a set of multicast destinations. We suppose that all multicast requests which we deal with can be transmitted in the network as a set of unicast transmissions (see Section~\ref{TreeConstr}). Therefore, we do not have to make precise the amount of data to transfer.  For a given multicast request $\epsilon$ we first determine its tree, $A_\epsilon = (V_{ A_\epsilon },E_{ A_\epsilon })$. This tree is a subgraph of $G$ rooted in $e$, whose leaves are in the set $R$ and whose arcs are directed from the root towards the leaves. We note $D_{A_\epsilon}$ the diffusing nodes in $A_\epsilon $,  $D_{A_\epsilon}\subseteq V_{A_\epsilon}$. Their allowed number is written as $k$. We now try to determine the choice of diffusing nodes in order to minimize the bandwidth consumption. 

We will adopt as a metric of the bandwidth used by a multicast request a total number of arcs which construct its tree taking into account the fact that an arc may transport the same data more than once. To define this metric formally we start by determining the situations in which a request $\epsilon$ is satisfied by a set of paths in the multicast tree, $\mathcal{S}(D_{A_\epsilon})$. These situations are as follows:
\begin{itemize}
\item every node of $R$ is the final extremity of exactly one path in $\mathcal{S}(D_{A_\epsilon})$,
\item every node of $D_{ A_\epsilon }$ is the final extremity of at most one path in $\mathcal{S}(D_{A_\epsilon})$,
\item the origin of a path in $\mathcal{S}(D_{A_\epsilon})$ is either $e$ or a node of $D_{ A_\epsilon }$; in the latter case, the node of $D_{ A_\epsilon }$ is also the final extremity of a path in $\mathcal{S}(D_{A_\epsilon})$,
\item any node of $a\in D_{ A_\epsilon }$ is in a path $p\in  \mathcal{S}(D_{A_\epsilon})$ only if it is the final extremity or the origin of $p$.
\end{itemize}
The metric $\mbox{load}_{ A_\epsilon }$ is defined as a sum of lengths of all paths in $\mathcal{S}(D_{A_\epsilon})$. The optimization problem which consists in placing $k$ diffusing nodes can be thus formalized as: 
\begin{definition}
{\bf Diffusing Nodes in Multicast Tree Problem (DNMTP)}\\
{\bf Data: } a directed connected symmetrical graph $G=(V,E)$, a multicast request $\epsilon = (e,R)$, a rooted multicast tree corresponding to this request $A_\epsilon$, and a natural number $k$.\\
{\bf Goal: } Find   $D_{ A_\epsilon }$, $|D_{ A_\epsilon }|\leq k$ so that $\mbox{load}_{ A_\epsilon }$ is minimal.
\end{definition}

\section{Properties of the solution induced by the subset of vertices $D$ \label{properties}}

\newcommand{\shadow}{window}
\newcommand{\border}{path number}

We now focus on a given multicast $\epsilon$ and we omit the subscript $\epsilon$ in the formul\ae{} for their clarity. This section is devoted to studying properties of the solution
induced by a set $D$ of diffusing nodes, $\mathcal{S}(D)$. We introduce the notation used for its
description. For any $u$, $u\in V_{A}$ in $A$
we define $A^u$ as a sub-tree of $A$ rooted in $u$.
We also define three parameters of $u$ in $A$.   A set $D^u$ is a set of
diffusing nodes in tree $A^u$ ($D^u\subseteq D$). A set
$R^u$ is a set of destinations nodes in tree $A^u$ ($R^u\subseteq R$). $a^u$ is the arc connecting $A^u$ from the
remainder of $A$. We propose:
\begin{definition}
  Let $D$ be a set of vertices in $A$. Let  $u$ be a vertex in $A$.
The \emph{path number} $\mbox{pn}(u)$ is a number of paths in a
solution  $\mathcal{S}(D)$ spanned on $A$ which pass through
$u$ or which terminate in $u$.  The  {\em \shadow}   of the solution $\mathcal{S}(D)$  on arc $a^u$ is 
an triplet of integers $(\beta, d,load)$ where  $\beta$ is its  {\border} $\mbox{pn}(u)$, $d=|D^u|$ and  $load$ represents the load of $\mathcal{S}(D)$   in tree $A^u$.
\end{definition}
We can notice that each solution induced by the set of diffusing
nodes $D$, can be defined by each \shadow\  for each arc of tree $A$.

\begin{lemma}  \label{lemma:1child}
   Let $D$ be a set of vertices in $A$.  Let $u$ be a node having one child $u_1$ of $A$. The    \shadow\ on arc $a^{u}$ is
equal to
\begin{equation}
  \label{eq:1}
  \left\{
    \begin{array}[]{llll}
      (1,d+1,load+1)  &   \mbox{\rm if} & u \in D  & \\
   (b+1,d,load+b+1)  &   \mbox{\rm if} & u  \notin D & \mbox{\rm and } u \in R\\
      (b,d,load+b)  &   \mbox{\rm if} & u \notin D & \mbox{\rm and } u \notin R\\
   
    \end{array}
\right.
\end{equation}
where the    \shadow\ on arc $a^{u_1}$ is $ (b,d,load) $.
\end{lemma}

\begin{proof}
First, we assume that $u \in D$ ($u$ is a diffusing node). By definition of {\border}, arc $a^u$ has a {\border} equal to one.
Since the    \shadow\ on arc $a^{u_1}$ is $ (b,d,load) $, tree $A^{u}$ contains one more diffusing node than tree $A^{u_1}$.

Second, we assume that $u \notin D$.  So, $u$ is not a diffusing
node and tree $A^{u}$ contains exactly the same set of diffusing nodes as   tree $A^{u_1}$. If $u \in R$, then $u$ is the final extremity of exactly one
path. So the {\border} on arc $a^u$ is equal to the {\border}  on arc $a^{u_1}$  plus one. 
 If $u \notin R$, then  the {\border} on arc $a^u$ is equal to the {\border}  on arc $a^{u_1}$.

 From these  statements, we can compute the load of the solution $\mathcal{S}$  induced by $D$ in tree $A^u$. The load increases by the {\border} on arc $a^u$.

\end{proof}


Now, using the same arguments in the proof of Lemma \ref{lemma:1child},  we extend it  when $u$ has several children.

\begin{lemma}  \label{lemma:Lchilds}
   Let $u$ be a node  of $A$ having $\ell$ children $u_1\dots,u_\ell$. The    \shadow\ on arc $a^{u}$ is
equal to
\begin{equation}
  \label{eq:2}
  \left\{
    \begin{array}[]{llllll}
      (1,&1+ \sum_i^\ell d_i,&\sum_i^\ell load_i+1)  &   \mbox{\rm if} & u \in D  & \\
      (1+ \sum_i^\ell b_i,& \sum_i^\ell d_i ,& \sum_i^\ell load_i+b+1 )  &   \mbox{\rm if} & u  \notin D & \mbox{ and } u \in R\\
      (\sum_i^\ell b_i, &\sum_i^\ell d_i,&\sum_i^\ell load_i +b  )  &   \mbox{\rm if} & u \notin D & \mbox{ and } u \notin R\\
    \end{array}
\right.
\end{equation}
where the    \shadow\ on arc $a^{u_i}$ is $ (b_i,d_i,load_i) $ for any $i$, $1\leq i \leq \ell$.
  
\end{lemma}
 
Now, we want to compare two solutions by  introducing   a partial order for each node.

\begin{definition}
  Let $D$ and $D'$ be two subsets of vertices in $A$. Let $v$ be a
  vertex.  $\mathcal{S}(D) \preceq_v \mathcal{S}(D')$ if and only if
  the following three conditions are  simultaneously satisfied: (i) $b
  \leq b'$; (ii) $d \leq d'$ (iii) $load\leq load'$, where the \shadow\
  of the solution   $\mathcal{S}(D)$ (respectively   $\mathcal{S}(D')$) on arc $a^{v}$ is $
  (b,d,load) $ (respectively $ (b',d',load') $).
\end{definition}

\begin{property} \label{prop:order}
  Let $u$ be a vertex in $A$. Let $D$ and $D'$ be two subsets of  vertices in  $A$ such that $\mathcal{S}(D) \preceq_u  \mathcal{S}(D')$. Then the solution induced by    $D''$, where $D''=(D'\setminus D'^u)   \cup  D^u$, satisfies the following property
  \begin{center}
     $\mathcal{S}(D'')  \preceq_v  \mathcal{S}(D')$ for all nodes $v$ not in $A^u$
  \end{center}
\end{property}
\begin{proof}
Let $P$ be the path between root $e$ and vertex $v$.
 
First, we focus on vertices $v$ outside $A^u$ and not in $P$. Since ${D''^v}= {D'^v}$, arc $a^v$ has the same \shadow\  of  the solution induced by $D$  and of the    solution   $ \mathcal{S}(D')$. 

Second, we focus on vertices $v$ in $P$. By definition of the partial order $ \preceq_u$, we have (i) $b \leq b'$, (ii) $d \leq d'$, and 
  (iii) $load \leq load' $, where the \shadow\ of  the solution $\mathcal{S}(D)$ (respectively   $\mathcal{S}(D')$)  on arc $a^{u}$ is $ (b,d,load) $ (respectively $ (b',d',load') $).  
Now, we can compute the \shadow\ of  the solution  $\mathcal{S}(D'')$ on arc $a^{t}$ where $t$ is the father of node $u$.
Let  $ (b'_t,d'_t,load'_t) $ be the \shadow\ of  the solution  $\mathcal{S}(D')$ on arc $a^{t}$.

 From Lemma \ref{lemma:Lchilds},
 if $t \in D'$, then 
the \shadow\ of  the solution $\mathcal{S}(D'')$   on arc $a^{t}$ is $ (1,d'_t-d'+d,load'_t-load'+load) $. Thus  $\mathcal{S}(D'')  \preceq_t  \mathcal{S}(D')$. We can apply the same arguments as previously for the other case. The same reasoning  goes for each vertex of this path starting from the father of $t$ until the root.
This completes the proof of Property \ref{prop:order}.
\end{proof}

\begin{definition}
  Let $u$ be a vertex in $A$. Let $D$ be a subset of vertices in
  $A$. $D$ is \textit{sub-optimal} for  $A^u$ if and only if for any
  $D'$ which is a subset of vertices in $A$ such that $d =   d'$ and $b =  b'$, we have
  $load \leq load'$ where the \shadow\ of the solution 
$\mathcal{S}(D)$ (respectively  $\mathcal{S}(D')$)
 on arc $a^{u}$ is $ (b,d,load) $ (respectively $
  (b',d',load') $).
\end{definition}

\begin{property} \label{prop:optimal} Let $u$ be a vertex of $A$ having
  $\ell$ children $u_1\dots,u_\ell$. Let $D$ be a subset of vertices
  in $A$. If $D$  is \textit{sub-optimal} for node $u$, then $D$   is also \textit{sub-optimal} for node $u_i$, for any integer $i$, $1\leq i \leq \ell$.

\end{property}

\begin{proof}
  We can prove this property by contradiction. Assume that there is at least  one child $u_i$
 of $u$ such that  $D$    is not {sub-optimal} for node $u_i$. So it implies that there exists a subset $D'$ such that  $D'$    is  {sub-optimal} for node $u_i$ and such that     $d=d'$ and $b=b'$, we have
  $load' < load$ where the \shadow\ of the solution $\mathcal{S}(D)$ (respectively  $\mathcal{S}(D')$) on arc $a^{u}$ is $ (b,d,load) $ (respectively $
  (b',d',load') $).  So, using Lemma~\ref{lemma:Lchilds}, we can construct a subset $D''$  such that   $D''=(D\setminus D^{u_i})   \cup  D'^{u_i}$ and such that      $\mathcal{S}(D'')  \preceq_u  \mathcal{S}(D')$.
So, it implies that $D$ is not optimal. So there is a contradiction.

\end{proof}

\section{Algorithm, its complexity and optimality \label{algo}}
Our algorithm is based on the dynamic approach. We introduce the notation used for its description. For any $u$, $u\in V_{A_\epsilon}$ in $A_\epsilon$ we define $A_\epsilon^u$ as a sub-tree of $A_\epsilon$ rooted in $u$. We also define two parameters of $u$ in $A_\epsilon$. The height $h(u)$ is a distance between $u$ and $e$ in $A_\epsilon$. We also note $h_{\max} = \max_{u\in V_{A_\epsilon}} h(u)$. The path number $\mbox{pn}(u)$ is a number of paths in a solution  $\mathcal{S}(D_{A_\epsilon})$ with a given set of diffusing nodes spanned on $A_\epsilon$  which pass through $u$ or which terminate in $u$. It is obvious that if $u$ is a branching node then $\mbox{pn}(u) = 1$. 

The idea of our algorithm is to compute for any  $u$, $u \in V(A)$,
some sub-optimal sets $D$ of diffusing nodes for $A^u$ where the
window of the solution $\mathcal{S}(D)$ on arc $a^u$ is $(b, d,
load)$. One set $D$ is constructed for any value $b$, $1\leq b\leq
|R^u|$, any value $d$, $0\leq d\leq k$.  As the reader might already
guess, a sub-optimal set $D$ for the root $e$ gives a solution to our
problem. We want therefore to find these sets starting from the leaves
and ending up in the root of $A$.  As $u$ may be or may not be a
diffusing node, we have to know how to compute the two sets  for both
the cases.
\begin{figure}
\begin{center}
\setlength{\parskip}{.2cm}
\setlength{\headsep}{-.3cm}
\setlength{\footskip}{-.3cm}
\fbox{
\parbox{14.0cm}{
\begin{tabbing}
\hspace{.2cm} \= \hspace{.2cm} \= \hspace{.2cm} \= \hspace{.2cm} \= \hspace{.2cm} \= \hspace{.2cm} \= \hspace{.2cm} \= \kill
\textbf{Procedure \texttt{Mat\_Vec\_Filling}} \\
$1.$ \textbf{If} $u$ is a leaf \textbf{then} attribute the ``unitary'' $M(u)$ and $L(u)$ to $u$ \textbf{endIf} \\
$2.$ \textbf{If} $u$ is not a leaf \textbf{then} \\ 
$3.$  \> \> choose arbitrarily $v$ which is one of the successors of $u$ in $A_{\epsilon}^{u}$; \\
$4.$ \> \> \texttt{First\_Succ\_Mat\_Vec(u,v)}; mark $v$; \\
$5.$ \> \> \textbf{While} there is a successor of $u$ in $A_{\epsilon}^{u}$ which has not be marked yet \textbf{do} \\
$6.$ \> \> \> choose arbitrarily  $w$ among the non-marked successors of $u$ in $A_{\epsilon}^{u}$; \\
$7.$ \> \> \> \texttt{Others\_Succ\_Mat\_Vec(u,w)}; mark $w$ \\
$8.$ \> \> \textbf{endWhile} \\
$9.$ \textbf{endIf} 
\end{tabbing} }
}
\end{center}
\caption{The procedure {\tt Mat\_Vec\_Filling} \label{Mat_Vec_Filling}}
\end{figure}

As $u$ may not be equipped with the branching property, the minimal load of the sub-optimal set for it should be stored in the matrix $M(u)$ whose rows are indexed by $\mbox{pn}(u)$ (these indices are $1,2,\ldots,|R|$) and whose columns are indexed by the number of diffusing nodes deployed in $A^u$ (these indices are $0,1,\ldots,k$). If a solution does not exist, the corresponding matrix element is equal to zero.

As $u$ may become a branching node, the minimal load of the sub-optimal set can be stored in a line vector $L(u)$ because the path number of a diffusing node is always equal to one. 

In a nutshell: $M_{i,j}(u) = \alpha \neq 0$ ($L_{i}(u) = \alpha \neq 0$, respectively) if and only if a sub-optimal set $D$ exists in $A^u$  having its {\shadow}   on arc $a^u$ equal to $(j,i,\alpha)$.  (respectively to  $(1,i,\alpha)$).
For computational reasons the destinations $u$, which are leaves of $A$, have ``unitary'' matrix and vector attributed: $M_{1,0}(u)=1$, $L_1(u)=1$ and all other elements are zero.

As we have said above, our algorithm to solve the DNMTP attributes to each node $u$ its $M(u)$ and $L(u)$ starting from the leaves whose height is $H=h_{\max}$ and performing the bottom-up scheme with $H=H-1$ until the root is reached ($H=0$). The attribution of $M(u)$ and $L(u)$  to $u$ is realised by the procedure \verb+Mat_Vec_Filling+ (Fig.~\ref{Mat_Vec_Filling}). This procedure takes a node $u$ and its corresponding sub-tree as data. Intuitively speaking, this is a modified breadth-first search~\cite{Knu97} in which one arbitrarily chosen successor, treated first, computes its matrix and vector (the \verb+First_Succ_Mat_Vec+ procedure) in a different way from its brothers (the \verb+Others_Succ_Mat_Vec+ procedure). The leaves have the ``unitary'' matrix and vector assigned.
\begin{figure}
\begin{center}
\setlength{\parskip}{.2cm}
\setlength{\headsep}{-.3cm}
\setlength{\footskip}{-.3cm}
\fbox{
\parbox{14.0cm}{
\begin{tabbing}
\hspace{.2cm} \= \hspace{.2cm} \= \hspace{.2cm} \= \hspace{.2cm} \= \hspace{.2cm} \= \hspace{.2cm} \= \hspace{.2cm} \= \kill
\textbf{Procedure \texttt{First\_Succ\_Mat\_Vec(u,v)}} \\ 
$1.$ $L_1 (u) = \min^+_j(M_{0,j}(v))$\\
$2.$ \textbf{ForAll} $i$ such that $0 < i \leq k$ \textbf{do} \\
$3.$ \> \> $L_i(u) = 1+ \min^+ (\min^+_j (M_{i-1,j}(v), L_{i-1}(v))$ \\
$4.$ \textbf{endForAll} \\
$5.$ \textbf{ForAll} $i$ such that $0 \leq i \leq k$ \textbf{do} \\
$6.$ \> \> \textbf{ForAll} $j$ such that $0 < i \leq |R|$ \textbf{do} \\

$7.$ \> \> \> \textbf{If} $j==1$ \textbf{then} $\mbox{\tt elT} = 1 + \min^+ (M_{i,1}(v),L_i(v))$\\
$8.$ \> \> \> \textbf{else} $\mbox{\tt elT} = j + M_{i,j}(v)$\\
$9.$ \> \> \> \textbf{endIf};\\
$10.$ \> \> \> \textbf{If} $u$ is destination of multicast $\epsilon$ \textbf{then} \\ 
$11.$ \> \> \> \> $M_{i,j+1}(u) = \mbox{\tt elT}+1$ \textbf{else} $M_{i,j}(u) = \mbox{\tt elT}$ \\ 
$12.$ \> \> \> \textbf{endIf} \\
$13.$ \> \> \textbf{endForAll} \\
$14.$\textbf{endForAll} 
\end{tabbing} 
}}
\end{center}
\caption{The procedure {\tt First\_Succ\_Mat\_Vec} \label{First_Succ_Mat_Vec}}
\end{figure}

The procedure \verb+First_Succ_Mat_Vec+ operates on a node $u$ and one of its successors $v$ for which $M(v)$ and $L(v)$ are already known as \verb+Mat_Vec_Filling+ follows a bottom-up approach (Fig.~\ref{First_Succ_Mat_Vec}). It uses the variable \verb+elT+ to store the non-zero elements in a column $i$ of $M(u)$ and $L(u)$. The procedure executes the function $\min^+$ whose two arguments are natural. It returns a minimum of these two values in exception of the case in which one of the arguments is zero. The  other positive argument is when returned. The main idea is based on the observation that the weight of the multicast tree in $A^v \cup \{u\}$ is equal to the multicast weight in $A^v$ incremented by the weight of reaching $u$ which is itself equal to $\mbox{pn}(u)$. Let us remind the reader that $\mbox{pn}(u) =1$ when $u$ is a diffusing node and $\mbox{pn}(u)$ is a matrix row index otherwise. 

\paragraph{Remark 1:} From Lemma \ref{lemma:1child} and Property \ref{prop:order}, we can deduce, that if $u$ has one child $v$ for any $i, \ j$, $0 \leq i \leq k$ and $1 \leq j \leq |R|$ 
\begin{itemize}
\item $L_{i}(u) = 1 + \min^+(L_{i-1}(v), \min^+\{M_{i-1,j}(v):j:1 \leq j \leq |R|\}$     

\item $M_{i,1}(u) = 1+  \min^+ (L_{i}(v), M_{i,1}(v) )  $ 
\item $M_{i,j}(u) = j +   M_{i,j'}(v)   $ where $j\neq 1$, and $j'=j-1$ if $u\in R$, otherwise $j'=j$     
\end{itemize}
The  procedure \verb+First_Succ_Mat_Vec+ computes the formul\ae{} here above.
\begin{figure}
\begin{center}
\setlength{\parskip}{.2cm}
\setlength{\headsep}{-.3cm}
\setlength{\footskip}{-.3cm}
\fbox{
\parbox{14.2cm}{
\begin{tabbing}
\hspace{.2cm} \= \hspace{.2cm} \= \hspace{.2cm} \= \hspace{.2cm} \= \hspace{.2cm} \= \hspace{.2cm} \= \hspace{.2cm} \= \hspace{.2cm} \= \hspace{.2cm} \= \hspace{.2cm} \= \hspace{.2cm} \= \hspace{.2cm} \= \kill
\textbf{Procedure \texttt{Others\_Succ\_Mat\_Vec(u,w)}} \\
$1.$ \textbf{ForAll} $i$ such that $0 < i \leq k$ \textbf{do} \\
$2.$ \> \> $\mbox{\tt elT} = \infty$; \\
$3.$ \> \> \textbf{ForAll} $(x,y)$ such that $(x,y) \in \{1,2,\ldots,k\} \times \{1,2,\ldots,k\}$ and  $x+y = i$ \textbf{do} \\
$4.$ \> \> \> \textbf{If} $(L_x(u) + L_y(w)) < \mbox{\tt elT}$ \textbf{then} $\mbox{\tt elT} = L_x(u)) + L_y(w)$ \textbf{endIf};\\
$5.$ \> \> \> \textbf{If} $(L_x(u) + \min^+_j(M_{y,j}(w))) < \mbox{\tt elT}$ \textbf{then} $\mbox{\tt elT} = L_x(u)) + \min^+_j(M_{y,j}(w))$\\
$6.$ \> \> \>  \textbf{endIf};\\
$7.$ \> \> \textbf{endForAll} \\
$8.$ \> \> $V'_i(u) = \mbox{\tt elT}$ \\
$9.$ \textbf{endForAll};\\ 
$10.$ \textbf{ForAll} $i$ such that $0 \leq i \leq k$ \textbf{do} \\
$11.$ \> \> $\mbox{\tt elT} = \infty$; \\
$12.$ \> \> \textbf{ForAll}  $(x,y)$ such that $(x,y) \in \{0,1,\ldots,k\} \times \{0,1,\ldots,k\}$ and  $x+y = i$ \textbf{do} \\
$13.$ \> \> \> \textbf{ForAll} $j$ such that $0<j\leq |R|$ \textbf{do} \\
$14.$ \> \> \> \> $\mbox{\tt elT} = \infty$; \\
$15.$ \> \> \> \> \textbf{ForAll}  $(a,b)$ such that $(a,b) \in \{0,1,\ldots,k\} \times \{0,1,\ldots,k\}$ and  $a+b = j$ \textbf{do} \\
$16.$ \> \> \> \> \> \textbf{If} $M_{x,a}(u) + M_{y,b}(u)+b < \mbox{\tt elT} $ \textbf{then} \\
$17.$ \> \> \> \> \> \>  $\mbox{\tt elT}= M_{x,a}(u) + M_{y,b}(u)+b $\\
$18.$ \> \> \> \> \> \textbf{endIf} \\
$19.$ \> \> \> \> \textbf{endForAll};\\
$20.$ \> \> \> \> $M'_{i,j}(u) = \min^+(\mbox{\tt elT}, M_{x,j-1}(u)+L_y(w)+1)$\\
$21.$ \> \> \> \textbf{endForAll}\\
$22.$ \> \> \textbf{endForAll}\\
$23.$ \textbf{endForAll}\\
$24.$ $L(u) \leftarrow L'(u)$; $M(u) \leftarrow M'(u)$;   
\end{tabbing} 
}
}
\end{center}
\caption{The procedure {\tt Others\_Succ\_Mat\_Vec} \label{Others_Succ_Mat_Vec}}
\end{figure}

In lines 1--4 $L(u)$ is computed for $u$ seen as a diffusing
node. On the $i^{\mbox{\scriptsize th}}$ step the smallest positive
weight is chosen between weights of its predecessor $v$ seen as a
diffusing and a non-diffusing node. These weights are taken for $v$
with one diffusing node less because $u$ itself is diffusing. This
weight is increased by the weight of reaching $u$ which is one as $u$
is diffusing.  Lines 5--9 fill up $M(u)$ when $u$ is seen as
non-diffusing. Line 7 treats the case in which only one path
passes through or terminates in $u$. The successor of $u$ can be either
a diffusing or non-diffusing node. Otherwise (line 8) its successor
has to be a non-diffusing node. The case in which $u$ is a destination
despite the fact that it is not a leaf in ${A}$ is treated in
lines 10--12 as the weight of the access to $u$ has to be added.

\verb+Others_Succ_Mat_Vec+
(Fig.~\ref{Others_Succ_Mat_Vec}) operates on a node $u$ and its
successors $w$ different from $v$ which has already been 
examined in \verb+First_Succ_Mat_Vec+. The procedure uses the variable
\verb+elT+ as \verb+First_Succ_Mat_Vec+ does. Furthermore, the
procedure makes use of the auxiliary variables $M'(u)$ and $L'(u)$ to
store the new values of $M(u)$ and $L(u)$ as the current elements of
$M(u)$ and $L(u)$ are still in use. The procedure
\verb+Others_Succ_Mat_Vec+ is built up on the same principle as the
previous one. Lines 1--9 treat the filling up of $L(u)$ and
lines 10--23 treat the filling up of $M(u)$. The important difference
consists in traversing all the couples $(x,y), x,y=1,2\ldots,k$ or
$x,y=0,1,\ldots,k$ such that $x+y=i$. It leads from the fact that this
time the weight of the multicast tree in $A^v \cup A^v \cup \{u\}$ is
equal to the sum of the multicast weights in $A^v \cup \{u\}$ and in
$A^w$ with the branching nodes deployed in both $A^v \cup \{u\}$
and $A^w$. The matrix computation also requires an appropriate path number
in order to determine the additional tree weight (lines 15--19).

\paragraph{Remark 2:} From Lemma \ref{lemma:Lchilds} and Property \ref{prop:optimal}, we can deduce, that if $u$ has $\ell$ children $u_1,\dots, u_\ell$ for any $i, \ j$, $0 \leq i \leq k$ and $1 \leq j \leq |R|$ 
{\small
\begin{itemize}
\item $L^{\ell}_{i}(u) =  \min^+ \{L^{\ell-1}_{i'}(u) +L_{i''}(u_\ell) , L^{\ell-1}_{i'}(u) + \min^+_{j:1 \leq j \leq |R|} M_{i'',j}(v): i'+i''=i\}$     
\item $M^{\ell}_{i,j}(u) = \min^+\left(\begin{array}{@{}ll@{}}
     \min^+ \{ M^{\ell-1}_{i',j'}(u) +  M_{i'',j''}(u_\ell) +j'' & : i'+i''=i \land \  j'+j''=j\}, \\
     \min^+ \{ M^{\ell-1}_{i',j-1}(u) +  L_{i''}(u_\ell) +1 & : i'+i''=i \}, \\
  \end{array}\right)$    
\end{itemize}
}
where $M^{f}(u)$ and   $L^{f}(u)$ correspond to the matrix and the vector computed by the algorithm for the sub-tree of $A^u$ where $u$ has only $f$ children  $u_1,\dots, u_f$.

The  procedure \verb+Others_Succ_Mat_Vec+ computes the formul\ae{} here above.



\begin{theorem}
The optimal set of diffusing nodes is obtained by the configuration associated to $min_{i:1\leq i \leq k }L_i(e)$.  Its complexity is  ${\cal O}(k^2 |R|^2 |V_{A}|)$.
\end{theorem}

\begin{proof}
From Remarks 1 and 2, we can deduce from any $u$ in $V(A)$,  the algorithm \verb+Mat_Vec_Mat_Filling+ computes vector $L(u)$ and matrix $M(u)$ such that $\forall b, \ 1 \leq b \leq |R|$,   $\forall d, \  0 \leq d \leq k$, thus there are two sub-optimal sets  of diffusing nodes: one has  load $L_d(u)$ and   the other has  load $M_{d,b}(u)$.
\end{proof}


\section{Numerical results \label{Results}}
Our algorithm determines the optimal localizations for $k$ diffusing nodes in a multicast tree which has already been created for a request $\epsilon = (e,R)$. As we have signalled in Section~\ref{TreeConstr} there are numerous methods of construction of these trees.  We selected two heuristic methods in order to observe their impact on the efficiency of our algorithm.   The first one establishes a shortest path (ShP) between $e$ and each $r\in R$. The corresponding multicast tree $A_\epsilon^{\mbox{\scriptsize ShP}}$ is a union of these shortest paths. The second method, which is based on the $2$-approximable solution of the Steiner tree problem proposed in~\cite{TakMat80}, gives $A_\epsilon^{\mbox{\scriptsize StT}}$ tree. This Steiner problem formalized in terms of multicast demand can be written as:
\begin{definition}
{\bf Steiner Tree Problem (StTP)}\\
{\bf Data: } a connected undirected graph $G=(V,E)$, a multicast request $\epsilon = (e,R)$, and a natural number $k$.\\
{\bf Question: } Does a rooted tree $A_\epsilon$ exist such that the number of its arc is less than or equal to $k$?
\end{definition}
The heuristic algorithm~\cite{TakMat80} leans on polynomial algorithms of a minimum-weight spanning tree~\cite{Kru56} and of a shortest path~\cite{Dij59} coupled.

To generate a graph of $200$ nodes we apply the Waxman model~\cite{Wax88} of BRITE~\cite{MedLak+01} (with default parameters). We estimate with the $5$\% precision at the significance level $\alpha=0.05$ the average weight of multicast tree as a function of the destination number for both the algorithms which construct a tree. For each number of destinations we choose uniformly in $V$ a multicast source $e$ and next, we select the destinations of this source according to the uniform distribution in $V-\{e\}$.  
\begin{figure}
 \begin{minipage}[b]{.46\linewidth}
\begin{center}
\includegraphics[height=1.75in,width=2.0in]{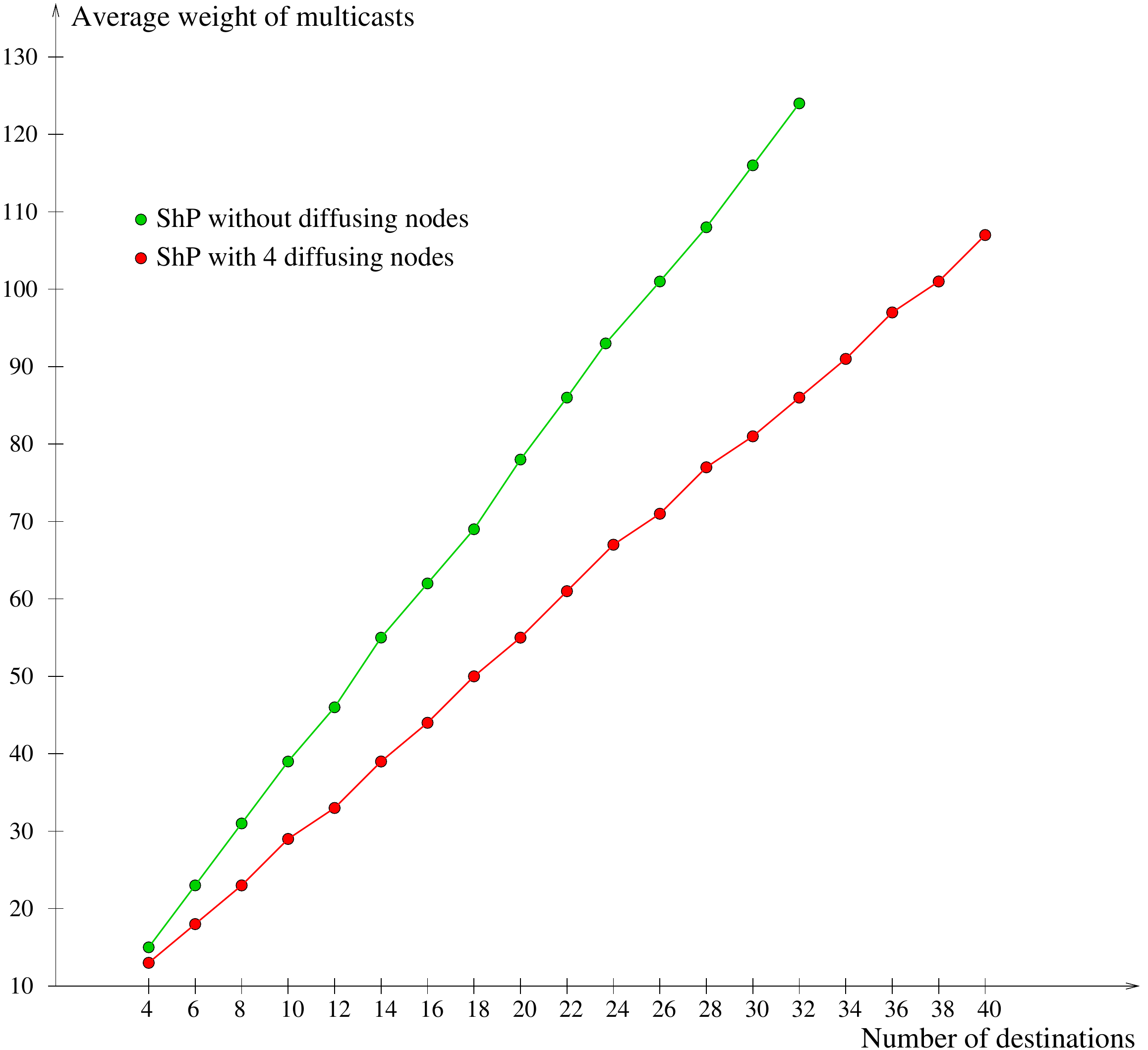}
\caption{Average $A_\epsilon^{\mbox{\scriptsize ShP}}$ weight as a function of the number of destinations with and without diffusing nodes\label{ShP1}}
\end{center}
 \end{minipage} \hfill
 \begin{minipage}[b]{.46\linewidth}
\begin{center}
\includegraphics[height=1.75in,width=2.0in]{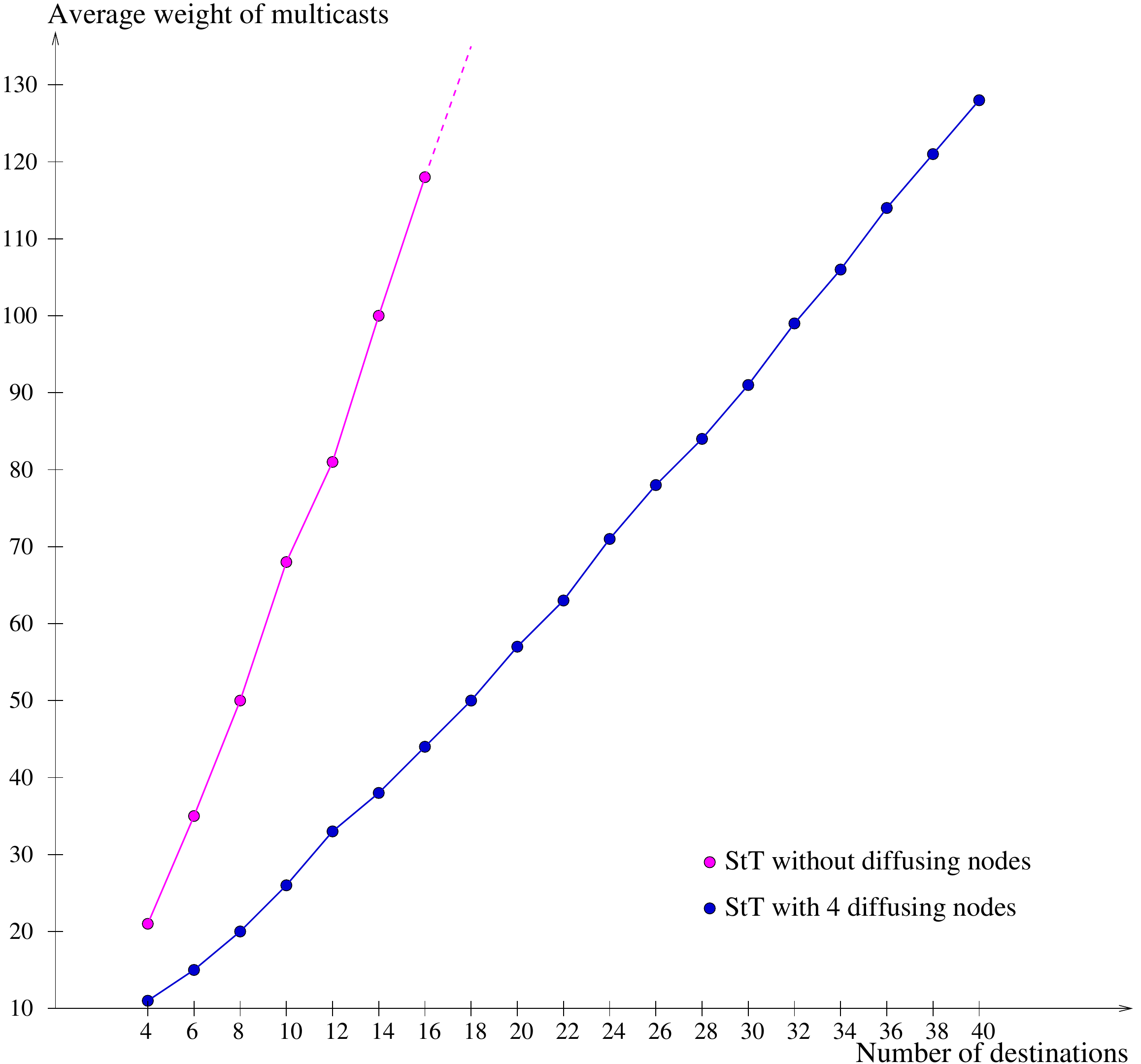}
\caption{Average $A_\epsilon^{\mbox{\scriptsize StT}}$ weight as a function of the number of destinations with and without diffusing nodes\label{StT1}}
\end{center}
 \end{minipage}
\end{figure}

 In order to perceive the impact of diffusing nodes on the tree weight we perform the computations with four nodes placed by our algorithm, and without them. In Fig.~\ref{ShP1} we observe that the weight reduction obtained for ShP with the diffusing nodes is significant (about $31$\% for $32$ destinations). The improvement obtained by the introduction of diffusing nodes into the trees built with StT (Fig.~\ref{StT1}) is even more substantial than in the previous case (about $65$\% for $16$ destinations). These two figures exhibit that ShP generates trees whose weight is less than those generated by StT. This fact is not astonishing as ShP always chooses a shortest path between the source and any destination.

In Fig.~\ref{StP1StT1}, in which  the relative difference between $A_\epsilon^{\mbox{\scriptsize ShP}}$  and $A_\epsilon^{\mbox{\scriptsize StT}}$ weights as StT tree weight percentage is depicted, we notice, however, that  this tendency is inverse for multicast trees with few destination (up to $18$). To explain this phenomenon we notice that 1) with a small number of destinations the shortest paths identified by ShP are disjoint, and 2) typically, the edges of a tree obtained by shortest paths are more numerous that those of a Steiner tree computed for an identical multicast demand.
\begin{figure}
\begin{center}
\includegraphics[height=1.25in,width=2.75in]{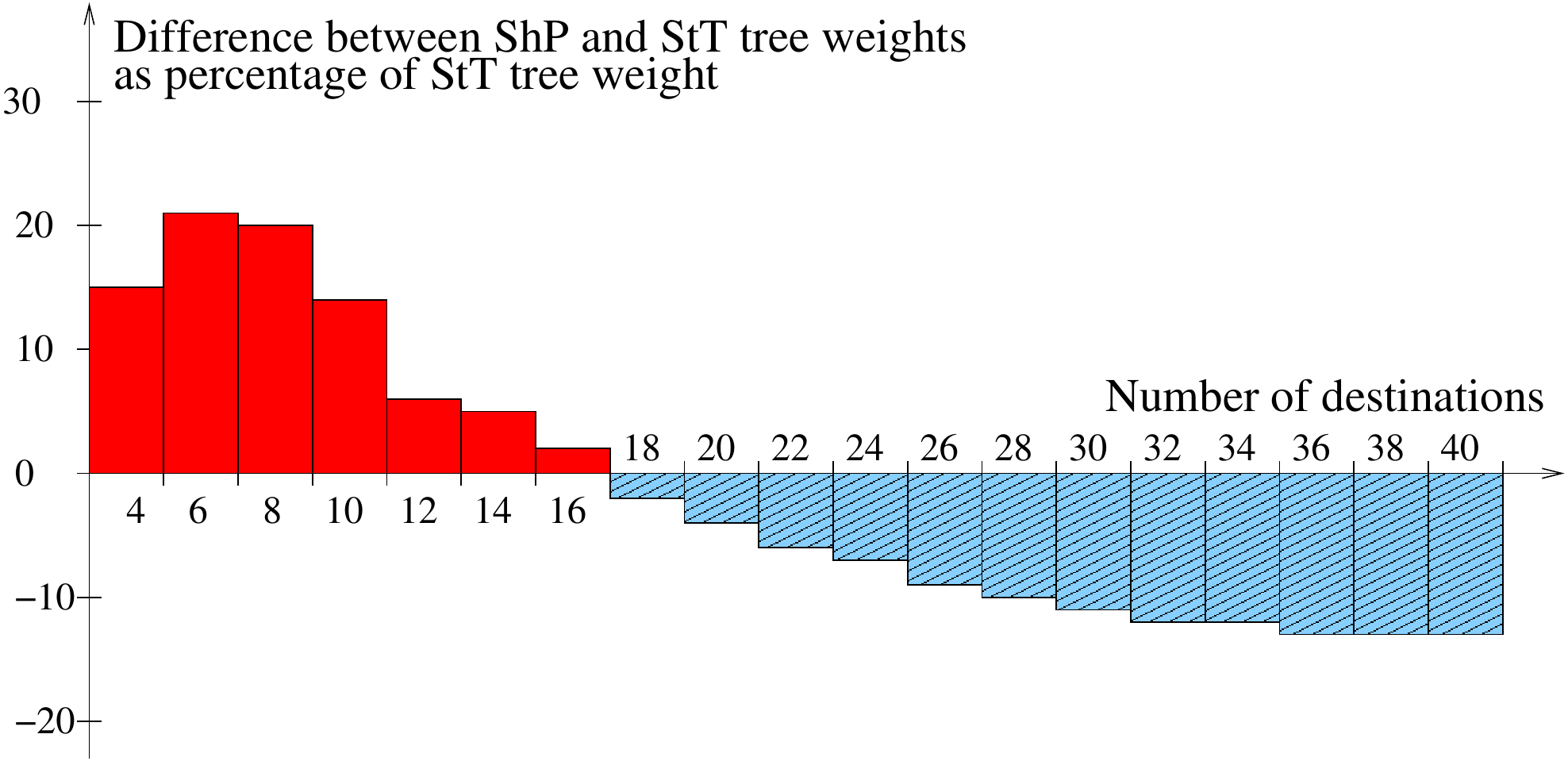}
\caption{Difference between weights of $A_\epsilon^{\mbox{\scriptsize ShP}}$ and $A_\epsilon^{\mbox{\scriptsize StT}}$ as $A_\epsilon^{\mbox{\scriptsize StT}}$ weight percentage with four diffusing nodes as a function of the number of destinations \label{StP1StT1}}
\end{center}
\end{figure}

We now fix the number of destinations to $20$ and we estimate the weights of trees obtained with ShP and StT algorithms in function of the number of branching nodes. We remind the reader that for $20$ destinations and $4$ diffusing nodes ShP turned out to be slightly more efficient than StT.  Figs.~\ref{ShP2} and~\ref{StT2} also show the average weight of ShP and StT trees estimated with the absence of diffusing nodes. In accordance with the comment made above in the context of the absence of diffusing nodes, ShP trees are almost twice as good as StT ones.
\begin{figure}
 \begin{minipage}[b]{.46\linewidth}
\begin{center}
\includegraphics[height=1.75in,width=2.0in]{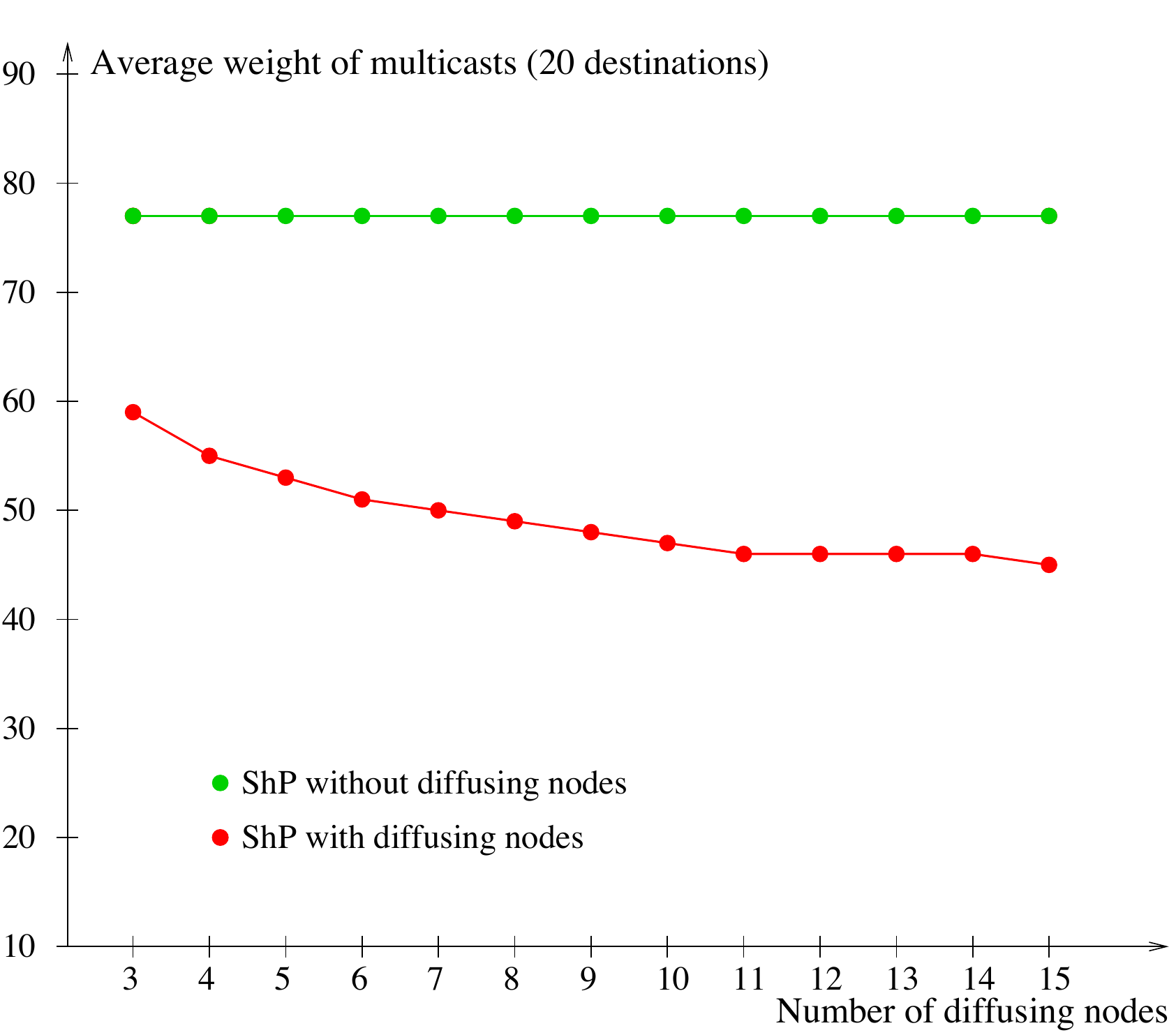}
\caption{Average $A_\epsilon^{\mbox{\scriptsize ShP}}$ weight as a function of the number of diffusing nodes for $20$ multicast destinations \label{ShP2}}
\end{center}
 \end{minipage} \hfill
 \begin{minipage}[b]{.46\linewidth}
\begin{center}
\includegraphics[height=1.75in,width=2.0in]{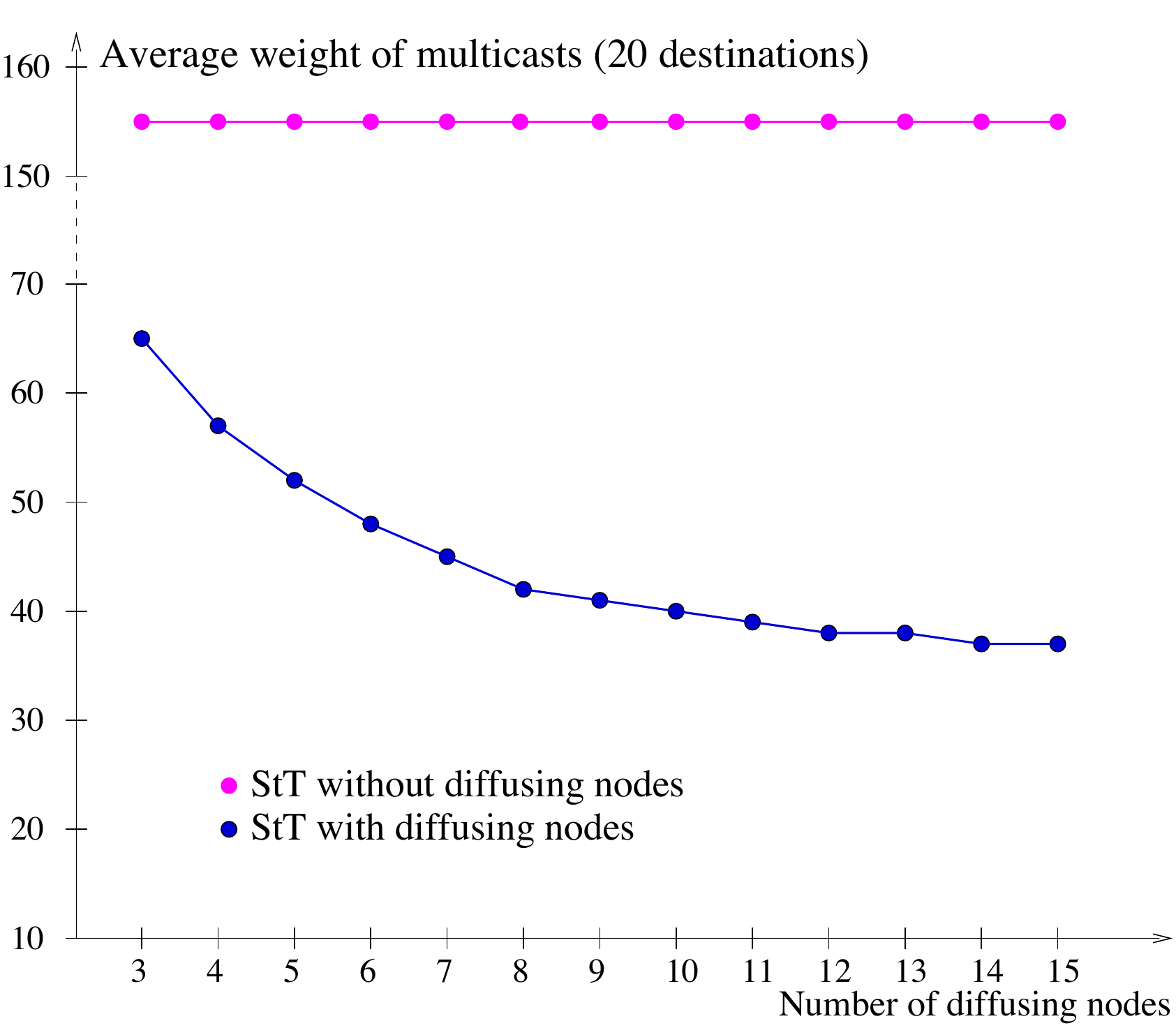}
\caption{Average $A_\epsilon^{\mbox{\scriptsize StT}}$  weight as a function of the number of diffusing nodes for $20$ multicast destinations  \label{StT2}}
\end{center}
 \end{minipage}
\end{figure}

The introduction of three diffusing nodes reduces the weight of $A_\epsilon^{\mbox{\scriptsize ShP}}$ by about  $20$\% (Fig.~\ref{ShP2}). Further additions allows one to lower the tree weight by almost $40$\% for $15$ branching nodes. The influence of the branching nodes on the reduction of the tree weight in the StT case is striking (Fig.~\ref{StT2}): an improvement of almost $60$\% in the case of three branching nodes until almost $75$\% for $15$ of them. Confronting the results of ShP and StT with diffusing nodes we observe that StT, despite its starting point at a worse position, reaches the tree weight of $40$ in the situation in which ShP has this weight of $48$.

In Fig.~\ref{StP2StT2} we observe the relative difference of ShP and StT tree weights for $20$ multicast destinations in function of the number of diffusing nodes. It is not surprising that for this relatively large number of destinations and few diffusing nodes StT exhibits better performance than ShP. When the number of branching nodes increases and approaches the number of destinations, ShP trees become lighter than StT ones for the same reasons as those mentioned in the comments on Fig.~\ref{StP1StT1}.

The next question we ask ourselves concerns the detection of the numbers of diffusing nodes and destinations up to which StT is more advantageous than ShP. For the network investigated above the critical point is $(4,18)$.  In Fig.~\ref{pointDestDiff} we mark critical points starting from which the ShP tree gives ``lighter'' solutions.   For the points above the line we recommend ShP (for example, for three branching nodes and $30$ destinations), for those below the line we recommend StT.  
\begin{figure}
\begin{center}
\includegraphics[height=1.25in,width=2.75in]{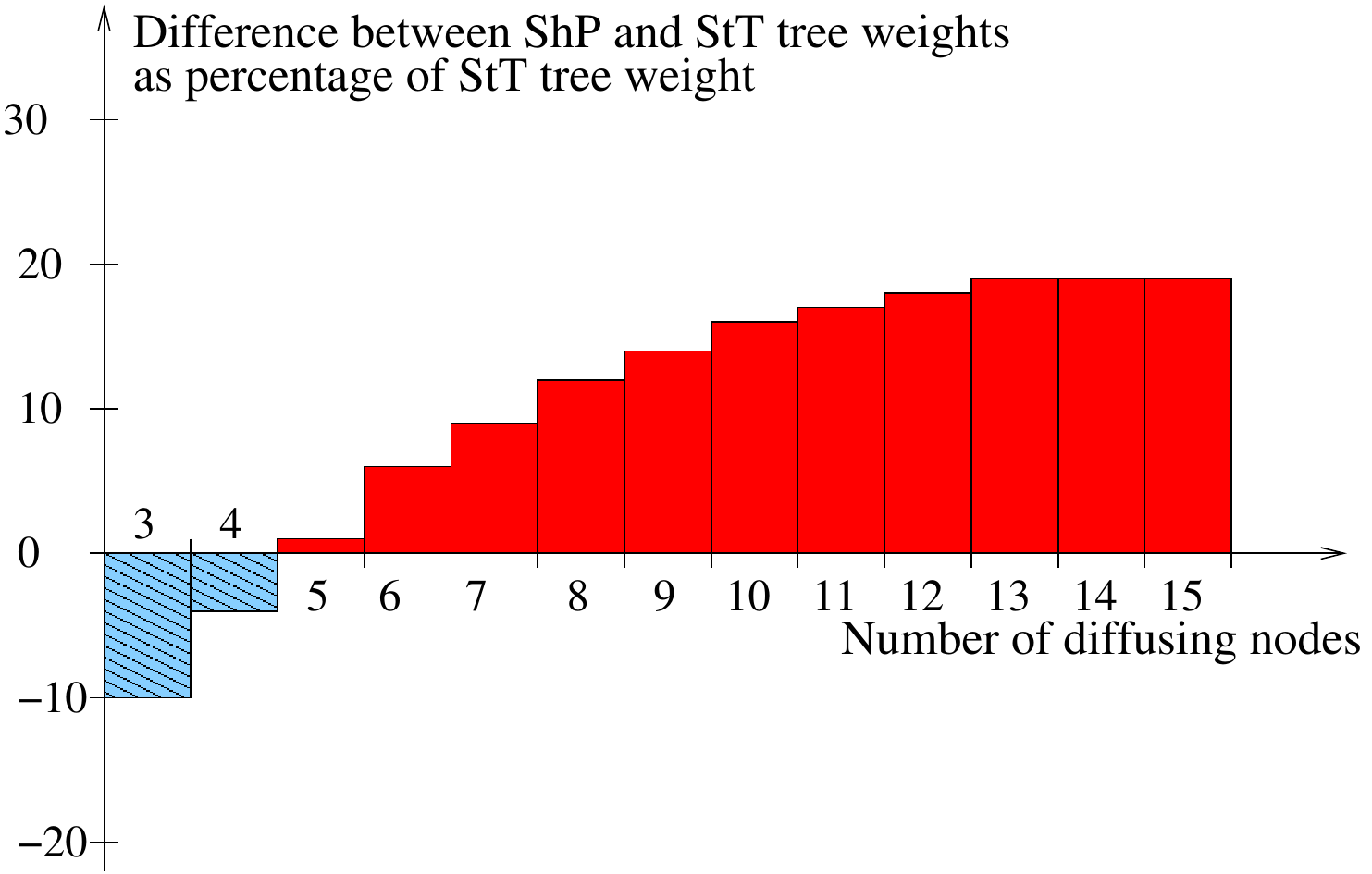}
\caption{Difference  between weights of ShP and StT trees as a percentage of StT tree weight as a function of the number of diffusing nodes for $20$ multicast destinations \label{StP2StT2}}
\end{center}
\end{figure}

As the critical points depicted in Fig.~\ref{pointDestDiff} form a straight line whose slope is five, we are now interested in what this gradient depends on. One may guess that it is determined by the average degree of the network. Indeed, if we look at Fig.~\ref{ratioDegree}, the gradient decreases as the average node degree increases. Consequently, the line seen in Fig.~\ref{pointDestDiff} inclines with the average degree growth. Therefore we conclude that StT is more favourable for loosely connected graphs and ShP is better for dense networks.

\begin{figure}
 \begin{minipage}[b]{.46\linewidth}
\begin{center}
\includegraphics[height=1.75in,width=2.0in]{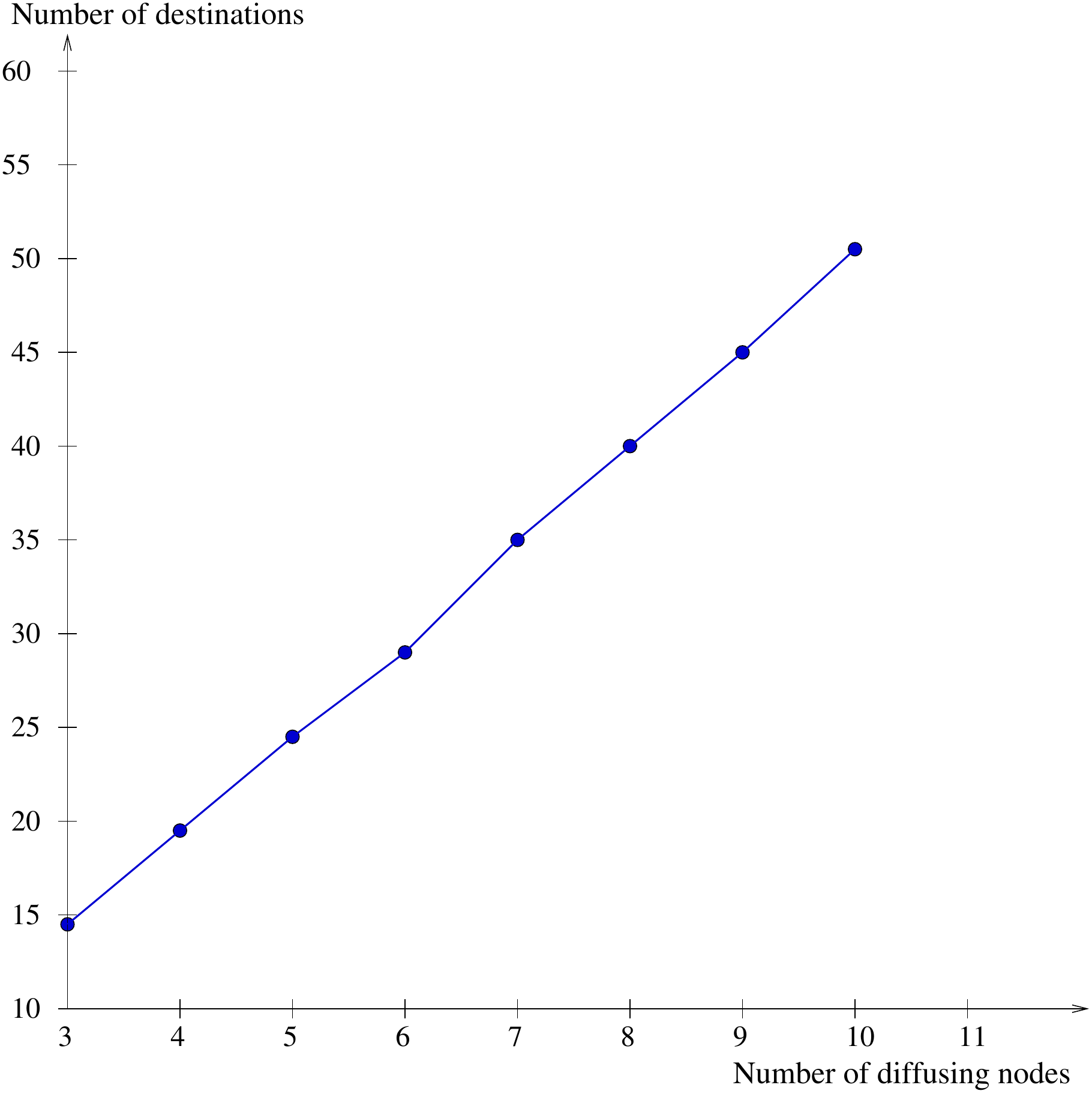}
\caption{Critical point line determining the utility of ShP and StT trees\label{pointDestDiff}}
\end{center}
 \end{minipage} \hfill
 \begin{minipage}[b]{.46\linewidth}
\begin{center}
\includegraphics[height=1.75in,width=2.0in]{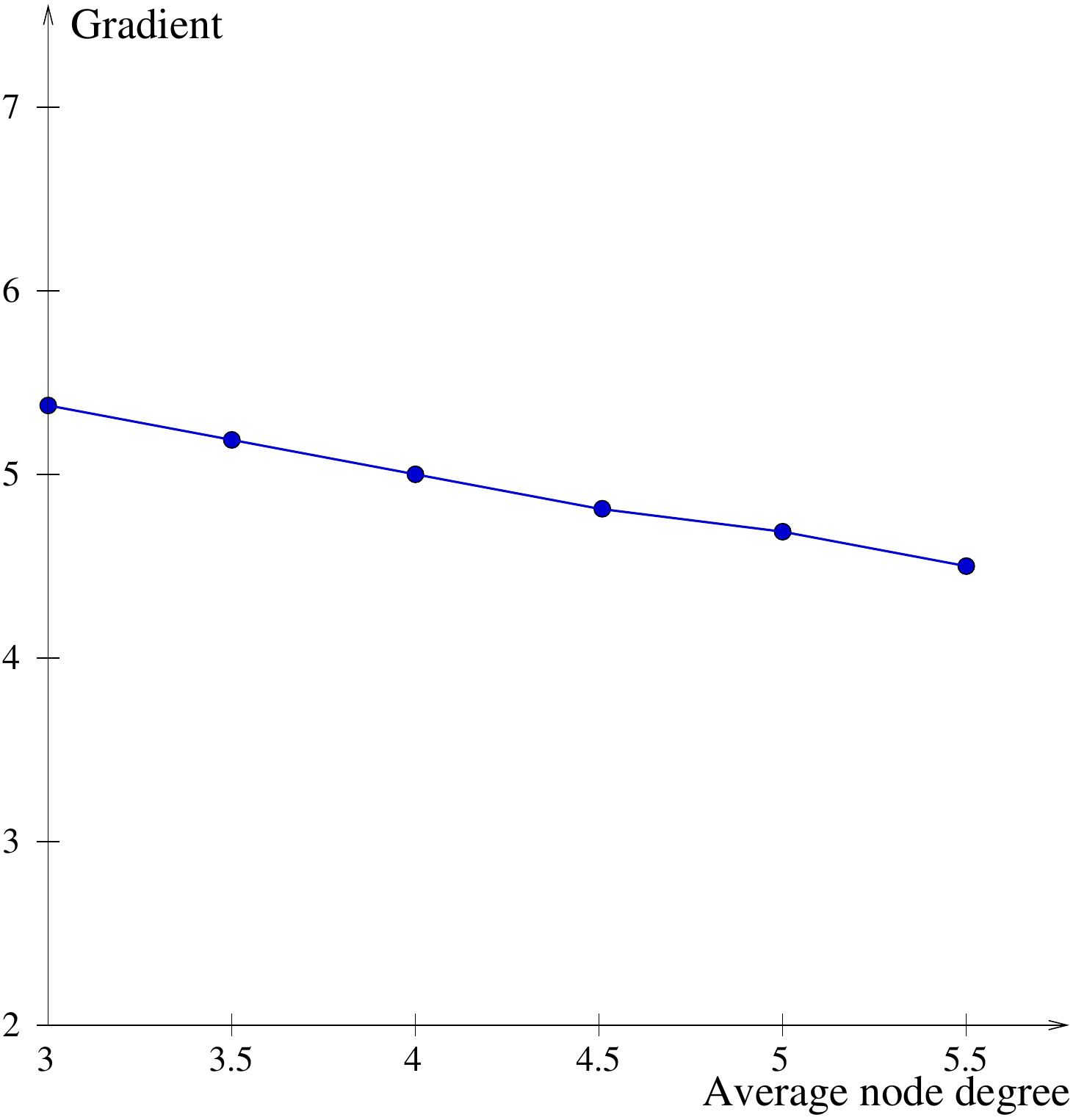}
\caption{Gradient of critical point line in function of the average graph degree \label{ratioDegree}}
\end{center}
 \end{minipage}
\end{figure}

\section{Conclusions and further work}
We studied a problem of infrastructural design of a commercial optical meshed network with a circuit switching routing protocol deployed. This problem was stated within the context of virtual services based on multicast transmission. It concerns frequent and voluminous multicast demands whose source and destinations are {\em \`a{} priori} known and its solution determines the localizations of branching nodes (i.e. routers with higher cost and energy consumption but which allow one to duplicate data and retransmit it in different directions).  A solution to this problem allows a network operator to use his available resources more efficiently and make more profit with less, or even without any, investment.

After formally stating the problem we proposed an algorithm to solve it. Next, we proved its optimality and computed its complexity which is polynomial.
We computed a gain in terms of the used bandwidth compared with multicast trees without any diffusing nodes. Among the two heuristic algorithms which we used to deploy multicast trees the first is based on the shortest past approach (ShP) and the second one exploits a solution to the Steiner tree problem (StT) in undirected graphs. We performed exhaustive computations in order to compare the efficiency of our algorithm for multicast trees built with ShP and StT. We observed the dependency of their efficiency on the numbers of diffusing nodes and destinations. This dependency is influenced by the average network degree. StT works better in loosely connected networks whereas ShP is more efficient for strongly connected ones.   

Generally speaking, we found ShP more efficient in finding a multicast tree than StT. We should not forget, however, that we used the $2$-approximable algorithm. It is not excluded that a more precise StT algorithm (for example \cite{ByrGra+10,RobZel00}) may give better results. We consider implementing these algorithms in order to verify their performance for our purposes.

We plan to continue this work in order to determine a specific solution in particular graphs (for example having bounded treewidth). We conjecture that our algorithm could be extended to this kind of graph.  On the other hand we consider pursuing our work on optimal multicast deployment by studying the Steiner problem in certain oriented graphs. 

\bibliographystyle{plain}
\bibliography{articleP2}

\begin{thebibliography}{10}

\bibitem{ArmFox+10}
M.~Armbrust, A.~Fox, R.~Griffith, A.~D. Joseph, R.~Katz, A.~Konwinski, G.~Lee,
  D.~Patterson, A.~Rabkin, I.~Stoica, and M.~Zaharia.
\newblock A view of {C}loud {C}omputing.
\newblock {\em Comm. {ACM}}, 53:50--58, April 2010.

\bibitem{Aud07}
O.~Audouin.
\newblock {CARRIOCAS} description and how it will require changes in the
  network to support {G}rids.
\newblock In {\em 20th Open Grid Forum}, 2007.

\bibitem{Bea89}
J.~Beasly.
\newblock An {SST}-based algorithm for the {S}teiner problem in graphs.
\newblock {\em Networks}, 19, 1989.

\bibitem{Ber66}
C.~Berge.
\newblock {\em The theory of graphs and its applications}.
\newblock Wiley, 1966.

\bibitem{BhaJaf83}
K.~Bharath-Kumar and J.~M. Jaffe.
\newblock Routing to multiple destinations in computer networks.
\newblock {\em {IEEE} {T}rans. {C}ommun.}, 31:343--351, March 1983.

\bibitem{BonChi+09}
E.~Bonetto, L.~Chiaraviglio, D.~Cuda, G.~Gavilanes Castillo, and F.~Neri.
\newblock Optical technologies can improve the energy efficiency of networks.
\newblock In {\em ECOC}, 2009.

\bibitem{ByrGra+10}
J.~Byrka, F.~Grandoni, T.~Rothvo\ss{}, and L.~Sanita.
\newblock An improved {LP}-based approximation for {S}teiner tree.
\newblock In {\em ACM-STOC}, 2010.

\bibitem{Dij59}
E.~W. Dijkstra.
\newblock A note on two problems in connexion with graphs.
\newblock {\em Numerische Mathematik}, 1:269--271, 1959.

\bibitem{FosKes+02}
I.~T. Foster, C.~Kesselman, J.~M. Nick, and S.~Tuecke.
\newblock Grid services for distributed system integration.
\newblock {\em {IEEE} Computer}, 35(6):37--46, 2002.

\bibitem{GarJoh79}
M.~R. Garey and D.~S. Johnson.
\newblock {\em Computers and Intractability: A Guide to the Theory of
  {NP}-Completeness}.
\newblock W. H. Freeman and company, 25th edition, 1979.

\bibitem{HwaRic92}
F.~K. Hwang and D.~S. Richards.
\newblock Steiner tree problems.
\newblock {\em Networks}, 22(1), 1992.

\bibitem{JeoQia+00}
M.~Jeong, C.~Qiao, Y.~Xiong, H.~C. Cankaya, and M.~Vandenhoute.
\newblock Efficient multicast schemes for optical burst-switched {WDM}
  networks.
\newblock In {\em ICC}, 2000.

\bibitem{Knu97}
D.~E. Knuth.
\newblock {\em The Art Of Computer Programming, vol. 1}.
\newblock Addison-Wesley, 1997.

\bibitem{KomPas+92}
V.~P. Kompella, J.~Pasquale, and G.~C. Polyzos.
\newblock Multicasting for multimedia applications.
\newblock In {\em INFOCOM}, 1992.

\bibitem{Kru56}
J.~B. Kruskal.
\newblock On the shortest spanning subtree of a graph and the traveling
  salesman problem.
\newblock {\em Proc. of AMS}, 7(1):48--50, February 1956.

\bibitem{MalZha+98}
R.~Malli, X.~Zhang, and C.~Qiao.
\newblock Benefit of multicasting in all-optical networks.
\newblock In {\em {SPIE} {A}ll {O}ptical {N}etworking}, pages 209--220, 1998.

\bibitem{Man04}
E.~Mannie.
\newblock {RFC} 3945 --- {GMPLS}, October 2004.

\bibitem{MedLak+01}
A.~Medina, A.~Lakhina, I.~Matta, and J.~Byers.
\newblock {BRITE}: An approach to universal topology generation.
\newblock In {\em MASCOTS}, Cincinnati, OH, USA, August 2001.

\bibitem{ReiTom+09}
V.~Reinhard, J.~Tomasik, D.~Barth, and M-A. Weisser.
\newblock Bandwith optimisation for multicast transmissions in virtual circuits
  networks.
\newblock In {\em {IFIP} Networking}, 2009.

\bibitem{RobZel00}
G.~Robins and A.~Zelikovsky.
\newblock Improved {S}teiner tree approximation in graphs.
\newblock In {\em ACM-SIAM}, 2000.

\bibitem{SalRee+97}
H.~F. Salama, D.~S. Reeves, and Y.~Viniotis.
\newblock Evaluation of multicast routing algorithms for real-time
  communication on high-speed networks.
\newblock {\em {IEEE} {J}. on Sel. Areas in Comm.}, 15(3):332--345, 1997.

\bibitem{TakMat80}
H.~Takahashi and A.~Matsuyama.
\newblock An approximate solution for the {S}teiner problem in graphs.
\newblock {\em Math. Jap.}, 24(6):573--577, 1980.

\bibitem{Ver+08}
D.~Verch\`e{}re, O.~Audouin, B.~Berde, A.~Chiosi, R.~Douville, H.~Pouylau,
  P.~Primet, M.~Pasin, S.~Soudan, D.~Barth, C.~Cader\'e{}, V.~Reinhard, and
  J.~Tomasik.
\newblock Automatic network services aligned with grid application requirements
  in {CARRIOCAS} project.
\newblock In {\em {G}rid{N}ets}, pages 196--205, 2008.

\bibitem{Wax88}
B.~M. Waxman.
\newblock Routing of multipoint connections.
\newblock {\em J-SAC}, 6(9), 1988.

\bibitem{ZhaWei+00}
X.~Zhang, J.~Y. Wei, and C.~Qiao.
\newblock Constrained multicast routing in {WDM} networks with sparse light
  splitting.
\newblock {\em J. of Lightwave Tech.}, 18(12):1917--1927, 2000.

\end{thebibliography}

\end{document}